\documentclass[a4paper,UKenglish,cleveref, autoref, thm-restate]{lipics-v2021}

\usepackage{mathtools}
\nolinenumbers

\def\ta{\mathtt{a}}

\DeclareMathOperator{\letters}{alph}
\DeclareMathOperator{\dist}{Dist}

\DeclarePairedDelimiter\floor{\lfloor}{\rfloor}

\newtheorem{problem}{Problem}

\title{Exploring Word-Representable Temporal Graphs} 

\author{Duncan Adamson}{School of Computer Science, University of St Andrews, United Kingdom}{duncan.adamson@st-andrews.ac.uk}{https://orcid.org/0000-0003-3343-2435}{}

\authorrunning{D. Adamson} 

\Copyright{Duncan Adamson} 

\ccsdesc[100]{\textcolor{red}{Replace ccsdesc macro with valid one}} 

\keywords{Temporal Graphs, Word-Representable Graphs} 

\category{} 

\relatedversion{} 

\EventEditors{John Q. Open and Joan R. Access}
\EventNoEds{2}
\EventLongTitle{42nd Conference on Very Important Topics (CVIT 2016)}
\EventShortTitle{CVIT 2016}
\EventAcronym{CVIT}
\EventYear{2016}
\EventDate{December 24--27, 2016}
\EventLocation{Little Whinging, United Kingdom}
\EventLogo{}
\SeriesVolume{42}
\ArticleNo{23}

\begin{document}

\maketitle

\begin{abstract}
Word-representable graphs are a subset of graphs that may be represented by a word $w$ over an alphabet composed of the vertices in the graph. In such graphs, an edge exists if and only if the occurrences of the corresponding vertices alternate in the word $w$. We generalise this notion to temporal graphs, constructing timesteps by partitioning the word into factors (contiguous subwords) such that no factor contains more than one copy of any given symbol. With this definition, we study the problem of \emph{exploration}, asking for the fastest schedule such that a given agent may explore all $n$ vertices of the graph. We show that if the corresponding temporal graph is connected in every timestep, we may explore the graph in $2\delta n$ timesteps, where $\delta$ is the lowest degree of any vertex in the graph. In general, we show that, for any temporal graph represented by a word of length at least $n(2dn + d)$, with a connected underlying graph, the full graph can be explored in $2 d n$ timesteps, where $d$ is the diameter of the graph. We show this is asymptotically optimal by providing a class of graphs of diameter $d$ requiring $\Omega(d n)$ timesteps to explore, for any $d \in [1, n]$.
\end{abstract}

\section{Introduction}

Word-Representable graphs, first introduced by Kitaev and Pyatkin \cite{JALC-2008-045}, are a subset of graphs that may be represented in a concise manner via a word (or string) of symbols. In these graph, the vertices are represented by the alphabet, while edges are defined by the structure of the word. See \cite{kitaev2017comprehensive} and \cite{kitaev2015words} for an overview of the main results on word-representable graphs.

Much work on word-representable graphs has focused on determining which such graphs are representable \cite{collins2017new,feng2024generalized,gaetz2020enumeration,glen2018representation,halldorsson2011alternation,srinivasan2024minimum}. Such work has shown that a large number of graph classes are representable by word graphs, including three colourable graphs, circle graphs, and crown graphs. As such, the study of problems on these graphs are non-trivial.

In this work, we extend the notion of word-representable graphs to \emph{temporal graphs}. Temporal graphs are a generalisation of static graphs, replacing the single edge set with an ordered sequence of sets. Each edge set is used to define a \emph{timestep} (also known as a snapshot), during which the only active edges are those in the corresponding timestep. Temporal graphs have attracted a large amount of attention in recent years, with a particular focus on reachability \cite{DeligkasP22,EMMZ21,meeks2022reducing} and exploration \cite{adamson2022faster,arrighi2023kernelizing,dogeas2023exploiting,ERLEBACH2021,erlebach2019two,ErlebachKDefficent,erlebach2023parameterised,michail2016traveling}.

\paragraph*{Our Contributions}

We present a set of results on the problem of exploring word-representable temporal graphs. This combines research on word-representable graphs, temporal graphs, and graph exploration. We provide two primary results. First, we show that for any word-representable temporal graph where each timestep is a fully connected graph, we can construct an exploration schedule requiring at most $2\delta n$ timesteps, where $\delta$ is the minimum degree of any vertex in the graph and $n$ the number of vertices. Second, we show that for any word-representable temporal graph where the underlying graph is connected and the word representation requires at least $n(2 d n + d)$ symbols, we can construct an exploration schedule requiring $2 d n$ timesteps, where $d$ is the diameter of the graph, and $n$ the number of vertices. Finally, we show that, for any number of vertices $n$ and diameter $d$, there exists a word-representable temporal graph with $n$ vertices and diameter $d$ requiring $\Omega(d n)$ timesteps to explore.


\paragraph*{Related Work}

While there is a large body of work on both word-representable graphs and exploration of temporal graphs, we highlight a small number of papers in both fields of particular relevance.

Regarding word-representable graphs, we point first to the work of Kitaev and Pyatkin \cite{JALC-2008-045} who introduced word-representable graphs as a tool for semigroup theory. From the prespective of structural results on word-representable graphs, in \cite{COURCELLE2008416}, Courcelle shows that circle graphs are word-representable. Halldorson et al. \cite{HALLDORSSON2016164} strengthen this by showing that comparability graphs and all $3$-colourable graphs a word-representable, as well as presenting communication from Limouzy that cover graphs are word-representable. On the negative side, Halldorson et al. show that it is NP-hard to determine if a given triangle free graph is word-representable.

Regarding the exploration of temporal graphs, we first mention the landmark paper by Erlebach et al. \cite{ERLEBACH2021}, which established many of the key ideas behind temporal graph exploration. In particular, they show a general $O(n^2)$ upper bound on the length of the fastest exploration schedule for a single agent on an always connected temporal graph. This is complemented by schedules of length $O(n^{1.5} k^{1.5} \log n)$ for graphs with treewidth $k$, $O(n^{1.8} \log n)$ for planar graphs, $O(n)$ for cycles with one chord, and $O(n \log^3 n)$ for $2 \times n$ grids. Building on this, Erlebach and Spooner showed in \cite{ErlebachKDefficent} that an always connected temporal graph in which each timestep is at most $k$-edges inactive can be explored in $O(k n \log n)$ timesteps. Adamson et al. \cite{adamson2022faster} further showed that $k$-chord graphs can be explored in $O(k n)$ timesteps, that graphs with treewidth $k$ can be explored in $O(k n^{1.5} \log n)$ timesteps, and that planar graphs can be explored in $O(n^{1.75} \log n)$ timesteps.

\section{Preliminaries}

Let  $\mathcal{N} = \{1,2,\ldots\}$ denote the natural numbers and set $\mathcal{N}_0 = \mathcal{N} \cup \{0\}$. Let $[i,n]=\{i, i+1, \ldots, n\}$ for all $i, n \in\mathcal{N}_0$ with $i \leq n$.

\paragraph*{Graphs and Temporal Graphs}

We define a \emph{graph} $G$ by a tuple containing a set $V$ of vertices, and set $E \subseteq V \times V$ of edges, each a pair of vertices. In this paper, we consider only \emph{undirected graphs}, by which we assume that $(v_i, v_j) = (v_j, v_i)$. When writing an edge explicitly as $(v_i, v_j)$ we call $v_i$ the \emph{start point} and $v_j$ the \emph{end point} of the edge. Given an arbitrary edge $e$ and vertex $v$, we use $v \in e$ to indicate that $v$ is one of the vertices in $e$.

We call two vertices $v_i, v_j \in V$ \emph{neighbours} if $(v_i, v_j) \in E$. We denote by $N(v)$ the set of neighbours of $v$, formally, $N(v) = \{v' \in V \mid (v, v') \in E\}$. The \emph{degree} of a vertex $v \in V$, denoted $d(v)$, is equal to the number of neighbours of $v$, formally $d(v) = \vert N(v) \vert$. A \emph{walk} in a graph $G$ is an ordered sequence of connected edges, $e_1, e_2, \dots, e_k$ such that the end point of the $i^{th}$ edge is the start point of the $(i + 1)^{th}$ edge. The start point of a walk is the start point of the first edge of the walk, and the end point of a walk is the end point of the last edge of the walk. Two vertices, $v_i$ and $v_j$, are \emph{connected} if there exists some walk starting at $v_i$ and ending at $v_j$. Note that a walk may contain multiple copies of the same edge. We denote the set of all walks in a graph $G$ by $\mathcal{W}(G)$. The \emph{length} of a walk $P = e_1, e_2, \dots, e_k$, denoted $\vert P \vert$ is the number of edges in the walk. A walk \emph{visits} a vertex $v$ if there exists at least one edge $e$ such that $v \in e$ The \emph{distance} between two vertices $v_i, v_j \in V$, denoted $\dist(v_i, v_j)$ is the walk of minimum length with $v_i$ as the start point and $v_j$ as the end point. Formally, letting $\mathcal{W}(G, v_i, v_j) = \{(e_1, e_2, \dots, e_{m}) \in E^* \mid  e_1 = (v_i, u_i), e_{m} = (u_j, v_j), u_i, u_j \in V\}$ be the set of walks in the graph $G$ starting at $v_i$ and ending at $v_j$, $\dist(v_i, v_j) = \min_ {p \in \mathcal{W}(G, v_i, v_j)} \vert p \vert$. We define the \emph{diameter} of a graph as $\max_{v_i, v_j \in V} \dist(v_i, v_j)$.

A \emph{temporal graph} $\mathcal{G}$ is a generalisation of a graph, defined over one set of vertices $V$, and $T$ sets of edges, $E_1, E_2, \dots, E_T$. We refer to the graph formed from the vertex set $V$ and the $t^{th}$ edge set, $E_t$, as the \emph{$t^{th}$ timestep}, denoted $G_t = (V, E_t)$. The \emph{lifetime} of a temporal graph $\mathcal{G} = (V, E_1, E_2, \dots, E_T)$ is equal to the number $T$ of edge sets in $\mathcal{G}$. The \emph{underlying graph} of a temporal graph $\mathcal{G} = (V, E_1, E_2, \dots, E_T)$, denoted $U(\mathcal{G})$, is the graph formed over the vertex set $V$ and the union of all edge sets, giving $U(\mathcal{G}) = (V, \bigcup_{t \in [1, T]} E_t)$.
A \emph{temporal walk} is an extension of a walk, with each step being an edge-index pair, $(e_1, t_1) (e_2, t_2) \dots (e_k, t_k)$ such that, for every $i \in [1, k - 1]$:
\begin{itemize}
    \item the edge $e_i$ is active in the timestep $E_{t_i}$,
    \item the end point of $e_i$ is the start point of $e_{i + 1}$, and,
    \item $t_{i + 1} > t_{i}$.
\end{itemize}
Additionally, $e_k$ must be active in timestep $t_k$.
Note that $t_{i + 1}$ may be greater than $t_{i} + 1$. In this case, assuming $e_i = (v_j, v_i)$, we say that an agent traversing this walk \emph{waits} at $v_i$ for $t_{i + 1} - t_i - 1$ timesteps. A temporal walk $P = (e_1, t_1), (e_2, t_2), \dots, (e_k, t_k)$ visits a vertex $v$ iff there is some edge $e_i$ in the walk such that $v \in e_i$. The \emph{length} of a temporal walk containing $k$ edges is equal to $t_k$, i.e. the timestep at which some agent following the walk would arrive at the end vertex. We say a temporal walk \emph{requires} $t$ timesteps if the walk has length $t$.

\begin{problem}[The Temporal Graph Exploration Problem]
    \label{prob:graph_exploration}
    Given a temporal graph $\mathcal{G} = (V, E_1, E_2, \dots, E_T)$, and start vertex $v \in V$ does there exist a temporal walk starting at $v$ and visiting every vertex in $V$?
\end{problem}

\paragraph*{Words}

An \emph{alphabet} $\Sigma=\{1,2,\ldots,\sigma\}$ is a finite set of symbols. A \emph{word} (also known as a \emph{string}) $w$ is a finite sequence of symbols from a given alphabet. We assume, for the remainder of this paper, that our alphabet $\Sigma$ is defined over some set of integers $1, 2, \dots, \sigma$.
The length of a word $w$, denoted $\vert w \vert$ is the number of letters in the word. For $i \in
[1, \vert w \vert]$ let $w[i]$ denote the $i^{th}$ letter of $w$. The set of all finite words over the alphabet $\Sigma$ is denoted by $\Sigma^{\ast}$. Given $n \in \mathcal{N}_0$, let $\Sigma^n$ denote all words in $\Sigma^{\ast}$ exactly of length $n$. Given two words $w, v$, we denote by $w v$ the word formed by the concatenation of $w$ and $v$, i.e. the word such that $w v[i] = w[i]$, if $i \in [1, \vert w \vert]$ or $v[i - \vert w \vert]$ if $i \in [\vert w \vert + 1, \vert w\vert + \vert v \vert ]$.  Given a word $w \in \Sigma^*$ and natural $k \in \mathbb{N}_0$, let $w^k$ denote the word formed by $k$ copies of $w$, satisfying $\vert w^k \vert = k \vert w \vert$ and $w^k [i] = w[i \bmod \vert w \vert]$, $\forall i \in [1, k \vert w \vert]$. Let $\letters(w) = \{\ta \in \Sigma \mid \exists i \in [\vert w \vert ]$ s.t. $w[i] = \ta \}$ be the alphabet of $w$.

Given a pair of words $u, w \in \Sigma^*$ where $\vert u \vert \leq \vert w \vert$, $u$ is a \emph{subsequence} of $w$ if there exists some set of indices $i_1, i_2, \dots, i_{\vert u \vert}$ such that $1 \leq i_1 < i_2 \dots < i_{\vert u \vert} \leq \vert w \vert$ such that $u = w[i_1] w[i_2] \dots w[i_{\vert u \vert}]$. Given a set of symbols $\mathcal{S} \subseteq \Sigma$ and word $w \in \Sigma^*$, we define the word $\pi_{\mathcal{S}}(w)$ recursively as:
\begin{itemize}
    \item the empty word $\varepsilon$ if $\vert w \vert = 0$, or,
    \item $\pi_{\mathcal{S}}(w[2, \vert w \vert])$ if $w[1] \notin \mathcal{S}$, or,
    \item $w[1] \pi_{\mathcal{S}}(w[2, \vert w \vert])$ if $w[1] \in \mathcal{S}$.
\end{itemize}
Note that this matches the notation given in \cite{Fleischmann2024}, though defined in an alternative manner. Informally, $\pi_{\mathcal{S}}(w)$ is the longest subsequence of $w$ composed of only those symbols in $\mathcal{S}$. See Figure \ref{fig:example_of_pi_function} for an example of this function.

\begin{figure}
    \centering
    \begin{tabular}{c|c}
        $\mathcal{S}$ & $\pi_{\mathcal{S}}(w)$ \\
        \hline
        $\mathcal{S} = \{a\}$ & $a a a$\\
        $\mathcal{S} = \{a, b\}$ & $ababab$\\
        $\mathcal{S} = \{a, c\}$ & $acaca$\\
        $\mathcal{S} = \{b, c\}$ & $b c b b$
    \end{tabular}
    \caption{Example of $\pi_{\mathcal{S}}$ on the word $w = a c b a c b a b$}
    \label{fig:example_of_pi_function}
\end{figure}

\paragraph*{Word-Representable (Temporal) Graphs}

Given a word $w \in \Sigma^*$, the graph represented by $w$, $G(w)$ is constructed as follows. Let $V = (v_1, v_2, \dots, v_n)$ be a set of $n$ vertices, each labelled uniquely by some symbol from $\Sigma = 1, 2, \dots, n$. Informally, $G(w)$ contains the edge $(v_x, v_y)$ iff the symbols $x$ and $y$ \emph{alternate} in $w$. A pair of symbols, $x, y \in \Sigma$ alternate in $w$ if $\pi_{\{x, y\}}(w) \in \{ (x y)^k, (x y)^k x, (y x)^k, (y x)^k y \mid k \in \mathbb{N}_0\}$.
Thus, we define the edge set $E$ as the set $\{(v_x, v_y) \mid \pi_{\{x, y\}}(w) \in \{ (x y)^k, (x y)^k x, (y x)^k, (y x)^k y \mid k \in \mathbb{N}_0\}\}$. A graph $G$ is \emph{word-representable} if there exists some word $w$ such that $G(w) = G$.

Given a word $w \in \Sigma^*$, the temporal graph represented by $w$, $TG(w)$, is constructed using $G(w)$ as a basis.
To determine the timesteps, we introduce the set of indices $S_1, S_2, \dots, S_T \in [1, \vert w \vert]$ as the set of \emph{start points} for each timestep. The value of each is computed recursively, with $S_1 = 1$, and $S_i$ the index such that $w[S_i] \in \letters(w[S_{i - 1}, S_i - 1])$ and, $\forall S_i' \in [S_{i - 1}, S_i - 1], w[S_i'] \notin \letters(w[S_{i - 1} + 1, S_{i' - 1}])$. Equivalently, $S_i$ is the index such that $\vert w[S_{i - 1}, S_i] \vert = \vert \letters(w[S_{i - 1}, S_i]) \vert + 1$.


With this set of indices, the timestep $t$ is defined with regard to the symbols appearing in the factor $w[S_t, S_{t + 1} - 1]$ (or, $w[S_T, \vert w \vert]$ for the final timestep). Formally, given some edge $(v_x, v_y)$ in the edge set $E$ of $G(w)$, the edge $(v_x, v_y) \in E_t$ iff either $x \in \letters(w[S_{t}, S_{t + 1} - 1])$ (resp., $x \in w[S_T, \vert w \vert]$) or $y \in \letters(w[S_{t}, S_{t + 1} - 1])$ (resp., $y \in w[S_T, \vert w \vert]$).
We call any factor of $w$ of the form $w[S_i, S_{i + 1} - 1]$ a \emph{timestep factor}.
We present an example in Figure \ref{fig:temporal_graph_example}.

\begin{figure}
    \centering
    \includegraphics[width=0.95\linewidth]{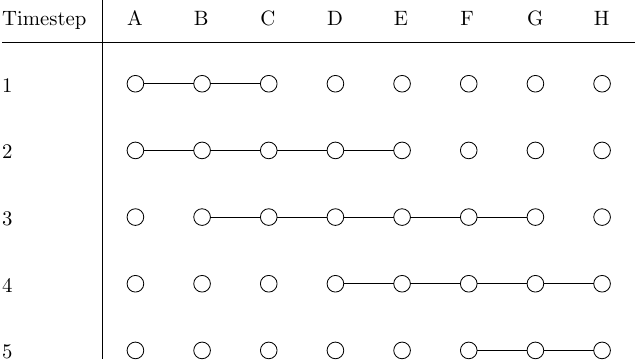}
    \caption{The temporal graph represented by $abacbdcedfegfhg$. Note that the start points are $S_1 = 1, S_2 = 3, S_4 = 7, S_5 = 11$.}
    \label{fig:temporal_graph_example}
\end{figure}

\section{Always Connected Graphs}
\label{sec:always_connected}

In this section, we provide results on exploring \emph{always connected graphs}. A temporal graph $\mathcal{G}$ is always connected if there exists, at every timestep, exactly one connected component containing all vertices in $\mathcal{G}$.

First, we provide some technical lemmas regarding the structure of any word $w \in \Sigma^*$ representing such a graph.

\begin{lemma}
    \label{lem:gap_between_letters}
    Let $\mathcal{G} = (V, E_1, E_2, \dots, E_{T})$ be a word-representable temporal graph represented by the word $w$, such that $\mathcal{G}$ is connected in every timestep. Then, given any vertex $v_x \in V$, $x \in \letters(w[S_t, S_{t + d(v_x) + 1} - 1])$, $\forall t \in [1, T - d(x) - 2]$, and $x \in \letters(w[S_{T - d(v_x) - 1}, \vert w \vert]$ where $S_1, S_2, \dots, S_T$ are the set of start points of $w$.
\end{lemma}

\begin{proof}
    Note that $\mathcal{G}$ is connected in every timestep if and only if there exists, for every $v_x \in V$ and $t \in [1, T]$, some edge $(v_x, v_y) \in E_t$. Further, as an edge $e$ is represented in $w$ by an subsequence of alternating symbols in $w$, the edge $e = (v_x, v_y)$ can appear in timestep $t$ if and only if $\{x, y \} \cap \letters(w[S_t, S_{t + 1} - 1]) \neq \emptyset$. As $v_x$ belongs to $d(v_x)$ edges, 
    By definition, there must exist in the factor $w[S_t, S_{t + 1} - 1]$ either the symbol $x$, or some $y$ such that $(v_x ,v_y) \in E$. Note that the edge $(v_x, v_y)$ can be in the graph $\mathcal{G}$ iff $\pi_{\{x, y\}}(w) \in \{(x y)^k, (x y)^k x, (y x)^k, (y x)^k y \mid k \in \mathbb{N}\}$. Therefore, there can be at most $d(v_x)$ contiguous timestep-factors that do not include $x$, giving the statement.
\end{proof}

\begin{corollary}
    \label{col:gap_between_edges}
    Let $\mathcal{G} = (V, E_1, E_2, \dots, E_{T})$ be a word-representable temporal graph represented by the word $w$, such that $\mathcal{G}$ is connected in every timestep. Then the edge $(v_x, v_y) \in E_{t} \cup E_{t + 1} \cup \dots \cup E_{t + \min(d(v_x), d(v_y))}$, $\forall t \in [1, T - \min(d(v_x), d(v_y))]$.
\end{corollary}

We now strengthen Lemma \ref{lem:gap_between_letters} and Corollary \ref{col:gap_between_edges} to prove that each edge $e$ in any word-representable always connected temporal graph appears at least once in every set of $\delta + 1$ contiguous timesteps. The main idea is to propagate the constraints given implicitly in Corollary \ref{col:gap_between_edges} across the graph. At a high level, we have that not only must some symbol $x$ satisfy $x \in \letters(w[S_t, S_{t + d(v_x) + 1} - 1]$, as given in Lemma \ref{lem:gap_between_letters}, but also $x \in \letters(w[S_t, S_{t + \min(d(v_x), d(v_y)) + 1} - 1]$ as indicated by Corollary \ref{col:gap_between_edges} and formally stated in Lemma \ref{lem:min_gaps}. More generally, however, $x$ must appear at least once between each occurrence of $y$, for any $y$ such that $v_y \in N(v_x)$. Thus, any bound on the number of occurrences of $y$ can be propagated through the graph to all other symbols. By extension, we get a bound on the number of timesteps in which an edge may be absent. Lemma \ref{lem:min_gaps} formalises this idea.

\begin{lemma}
    \label{lem:min_gaps}
    Let $\mathcal{G} = (V, E_1, E_2, \dots, E_{T})$ be a word-representable temporal graph represented by the word $w$, such that $\mathcal{G}$ is connected in every timestep. Then the edge $(v_x, v_y) \in E_{t} \cup E_{t + 1} \cup \dots \cup E_{t + \delta}$, $\forall t \in [1, T - \delta]$, where $\delta = \min_{v \in V} d(v)$.
\end{lemma}

\begin{proof}
    Consider some vertex $v_x$. Observe that, following Corollary \ref{col:gap_between_edges}, $(v_x, v_y)$ must be active at least once every $\min(d(v_x), d(v_y))$ timesteps. Therefore, $x, y \in \letters(w[S_t$, \\$S_{t + \min(d(v_x), d(v_y)) + 1} - 1])$. Now, consider some $z$ such that $v_z \in N(v_y)$ From Corollary \ref{col:gap_between_edges} we have that $y \in \letters(w[S_t$,$S_{t + \min(d(v_y), d(v_z)) + 1} - 1])$. Thus, $y \in \letters(w[S_t$,\\ $S_{t + \min(d(v_x), d(v_y), d(v_z)) + 1} - 1])$. Now, assume that $d(v_z) < \min(d(v_x), d(v_y))$. Therefore, as $y \in \letters(w[S_t, S_{t + d(v_z) + 1} - 1]), \forall t \in [1, T - d(v_z)]$, we have that $x \in \letters(w[S_t, S_{t + d(v_z) + 1} - 1]), \forall t \in [1, T - d(v_z)]$. Extending this argument across the graph shows that, $\forall x \in \Sigma$, $x \in \letters(w[S_t, S_{t + \delta + 1} - 1])$, and thus, for any $v_y \in N(v_x)$, $(v_x, v_y) \in E_{t} \cup E_{t + 1} \cup \dots \cup E_{t + \delta}$, $\forall t \in [1, T - \delta]$.
\end{proof}


We now present our main theorem for this section, showing that any word-representable always connected temporal graph may be explored in $2\delta n$ timesteps. The high-level idea behind this proof is to first construct a spanning tree on the underlying graph, then use a walk $P$ exploring this tree to derive a temporal walk exploring the temporal graph. We do so by following the same set of edges as $P$, waiting at each vertex at most $\delta$ timesteps for the next edge to become available. As $P$ can contain at most $2n$ edges, our temporal walk requires at most $2 \delta n$ timesteps, giving the algorithmic result.

\begin{theorem}
    \label{thm:time_complexity_of_exploring_always_connected_graphs}
    Let  $\mathcal{G} = (V, E_1, E_2, \dots, E_{T})$ be a word-representable temporal graph represented by the word $w$, such that $G$ is connected in every timestep. Then, $\mathcal{G}$ can be explored in $2\delta \vert V \vert$ timesteps starting at any vertex $v_x$, where $\delta = \min_{v \in V} d(v)$ is the minimum degree of any vertex in $\mathcal{G}$.
\end{theorem}

\begin{proof}
    Let $G$ be the underlying graph of $\mathcal{G}$, and let $\mathcal{T}$ be a minimum weight spanning tree of $G$. Observe that a walk $\tau = (e_1, e_2, \dots, e_m)$ exploring all the vertices of $\mathcal{T}$, and thus $\mathcal{G}$, starting at $v_x$ can be determined requiring at most $2n$ edges. We assume, for simplicity, that $e_i$ simply represents the $i^{th}$ edge taken in $\tau$, and thus there may exist some pair $e_i, e_j$ where $i < j$ and $e_i = e_j$.

    We construct a temporal walk exploring $\mathcal{G}$ from $\tau$ as follows. In the first timestep, the agent either moves along the edge $e_1$, if this edge is active, or waits at vertex $v_x$ until $e_1$ becomes active otherwise. In general, after traversing edge $e_i$ and arriving at vertex $v_y$, the agent will wait at $v_y$ until $e_{i + 1}$ becomes active, at which point the agent will move along $e_{i + 1}$.
    To determine the length of this temporal walk, observe that, by Lemma \ref{lem:min_gaps}, the agent must wait at the vertex $v_y$ for at most $\delta$ timesteps for the edge $(v_y, v_z)$ to become active, for any $v_y \in V$, $v_z \in N(v_i)$. Thus, the temporal walk requires $2\delta \vert V \vert$ timesteps.
\end{proof}

\section{General Graphs}

We now consider general word-representable temporal graphs. We assume, without loss of generality, that all graphs in the section contain a single component in the underlying graph, noting that any graph that does not satisfy this can not be explored. As in Section \ref{sec:always_connected}, we start with some structural results on the word representations of these graphs.

\begin{lemma}
    \label{lem:gaps_adjacent}
    Given a word-representable temporal graph $\mathcal{G}$ represented by the word $w$ containing the edge $(v_x, v_y)$. Then, $\vert \pi_{\{x\}} \vert - 1 \leq \vert \pi_{\{y\}} \vert \leq \vert \pi_{\{x\}} \vert + 1$.
\end{lemma}

\begin{proof}
    Observe that $\pi_{\{x, y\}} \in \{(x y)^k, (x y)^k x, (y, x)^k, (y, x)^k y\}$. Therefore, there can be at most one more copy of the symbol $x$ (resp. the symbol $y$) than $y$ (resp. $x$). 
\end{proof}

\begin{lemma}
    \label{lem:gaps_further}
    Let $v_x, v_y \in V$ be a pair of vertices such that $\dist(v_x, v_y) = d$. Then, $\vert \pi_{\{x\}}(w) \vert - d \leq \vert \pi_{\{x\}}(w) \vert \leq \vert \pi_{\{x\}} \vert + d$.
\end{lemma}

\begin{proof}
    We prove this by induction. Note that the case where $d = 1$ is given by Lemma \ref{lem:gaps_adjacent}. Assume now that, for every $d' \in [1, d - 1]$, this statement holds, and consider some pair $v_x, v_y \in V$ where $\dist(v_x, v_y) = d$. Let $v_z \in V$ be some vertex where $\dist(v_x, v_z) = d - 1$ and $(v_y, v_z) \in E$. Then, by the above assumption we have that $\vert \pi_{\{x\}}(w) \vert - (d - 1) \leq \vert \pi_{\{z\}}(w) \vert \leq \vert \pi_{\{x\}}(w) \vert + (d - 1)$. Further, by Lemma \ref{lem:gaps_adjacent}, we have $\vert \pi_{\{y\}}(w) \vert - 1 \leq \vert \pi_{\{z\}}(w) \vert \leq \vert p_{y}(w) \vert + 1$ and thus $\vert \pi_{\{x\}}(w) \vert - d \leq \vert \pi_{\{y\}}(w) \vert \leq \pi_{\{x\}}(w) \vert \leq \pi_{\{x\}}(w) \vert + d$.
\end{proof}

\begin{lemma}
	\label{lem:gaps_betwixt_symbols}
	Let $v_x, v_y \in V$ be a pair of vertices such that $\dist(v_x, v_y) = d$. Further, let $\chi_{i}$ be the index of the $i^{th}$ occurrence of $x$ in $w$, and let $\gamma_j$ be the index of the $j^{th}$ occurrence of $y$ in $w$. Then $\chi_{i -d} \leq \gamma_i \leq \chi_{i + d}$.
\end{lemma}

\begin{proof}
	We prove this inductively. First, consider the case where $d = 1$. Then, as $(v_x, v_y) \in E$, we have that $\pi_{\{x, y\}} \in \{(x y)^k, (x y)^k x, (y x)^k, (y, x)^k y \mid k \in \mathbb{N} \}$, the statement holds directly. Now, assume this holds for every value $d' \leq d - 1$ and let $v_x, v_y, v_z$ be some set of  vertices such that $\dist(v_x, v_z) = d - 1$, $(v_y, v_z) \in E$, and $\dist(v_x, v_y) = d$. Further, let $\zeta_i$ be the index of the $i^{th}$ occurrence of $z$ in $w$. Then, as $\chi_{i - (d - 1)} \leq  \zeta_i \leq \chi_{i + (d - 1)}$, and $\gamma_{i - 1} \leq \zeta_i \leq \gamma_{i + 1}$, we have $\chi_{i - d} \leq \gamma_i \leq \chi_{i + d}$.
\end{proof}

\begin{corollary}
    \label{col:everything_in_the_first_d_steps}
    Given a word-representable temporal graph $\mathcal{G}$, represented by the word $w$, $U(\mathcal{G}) = (V, \bigcup_{t \in [1, d]} E_t)$, where $d$ is the diameter of $U(\mathcal{G})$.
\end{corollary}


From Lemmas \ref{lem:gaps_adjacent}, \ref{lem:gaps_further}, \ref{lem:gaps_betwixt_symbols} and Corollary \ref{col:everything_in_the_first_d_steps}, we have the main tools used to build our algorithmic result. At a high level, we use a similar technique to Theorem \ref{thm:time_complexity_of_exploring_always_connected_graphs}, building a temporal walk exploring the graph from a walk exploring the spanning tree of the underlying graph. Lemma \ref{lem:general_graphs_waiting_time} and Corollary \ref{col:general_graphs_waiting_time_full} provide the last needed results for the approach, showing that each edge in the graph is inactive for at most $d$ consecutive timesteps.

\begin{lemma}
    \label{lem:general_graphs_waiting_time}
    Given a word-representable temporal graph $\mathcal{G}$, represented by the word $w$ such that $U(G)$ has a diameter of $d$. Then, given any edge $(v_x, v_j)$ active in timestep $t \in [1, T - d - 1]$, $\exists t' \in [t + 1, t + d + 1]$ such that $(v_x, v_y) \in E_{t'}$.
\end{lemma}

\begin{proof}
    Assume, for the sake of contradiction, that we have some edge $(v_x, v_y)$ such that $(v_x, v_y) \in E_t$ and, $\forall t' \in [t + 1, t + d]$, $(v_x, v_y) \notin T_t'$. Now, let $\xi \in [S_{t}, S_{t + 1} - 1]$ be the position in $w$ such that $w[\xi] \in \{x, y\}$ and, $\forall \xi' \in [\xi + 1, S_{t + 1} - 1]$, $w[\xi'] \notin \{x, y\}$. By extension, $\forall i \in [\xi + 1, S_{t + d + 1}]$, $w[i] \notin \{x, y\}$. Then, there must exist some symbol $z \in \Sigma$ such that $\vert \pi_{\{z\}}(w[S_t, S_{t + d + 1} - 1]) \vert = d$. However, by Lemma \ref{lem:gaps_betwixt_symbols}, we have a contradiction. Thus, the statement holds.
\end{proof}

\begin{corollary}
    \label{col:general_graphs_waiting_time_full}
    Given a word-representable temporal graph $\mathcal{G}$, represented by the word $w$ such that the underlying graph $U(\mathcal{G})$ has a diameter of $d$. Then, for any $t \in [1, T - d]$, $U(\mathcal{G}) = (V, \bigcup_{d' \in [0, d]}E_{t + d'})$.
\end{corollary}

\begin{proof}
    Follows from Corollary \ref{col:everything_in_the_first_d_steps} and Lemma \ref{lem:general_graphs_waiting_time}.
\end{proof}

From Lemma \ref{lem:general_graphs_waiting_time} and Corollary \ref{col:general_graphs_waiting_time_full}, we main now apply a similar technique to Theorem \ref{thm:time_complexity_of_exploring_always_connected_graphs} to derive an algorithm for exploring word-representable temporal graphs of length at least $n(2 d n + d)$ in $2 d n$ timesteps.

\begin{theorem}
    \label{thm:exploration_upper_bounds}
    Let $\mathcal{G} = (V, E_1, E_2, \dots, E_T)$ be a word-representable temporal graph, represented by the word $w$ such that $\vert w \vert \geq n(2 d n + d)$, where $d$ is the diameter of the underlying graph $U(\mathcal{G}) = (V, \bigcup_{t \in [T]} E_t)$. Then, $\mathcal{G}$ can be explored in at most $2 d n$ timesteps.
\end{theorem}

\begin{proof}
	First, note that we must have at least one symbol $x \in \Sigma$ such that $\vert \pi_{\{x\}}(w) \vert \geq 2 d n + d$. Thus, by Lemma \ref{lem:gaps_further}, we have that every symbol $y \in \Sigma$ satisfies $\vert \pi_{\{y\}}(w) \vert \geq 2 n d$. Further, we have $T \geq 2 n d + d$.
    
    Now, let $\mathcal{T}$ be a spanning tree of $U(\mathcal{G})$, and let $P = e_1, e_2, \dots, e_m$ be a walk exploring $\mathcal{T}$ starting from some arbitrary vertex $v \in V$. We construct our temporal walk, $Q$, exploring $\mathcal{G}$ as follows. At timestep 1, the agent attempts to traverse edge $e_1$ if active. Otherwise, the agent waits at the stating vertex until $e_1$ becomes active. Note that by Lemma \ref{lem:general_graphs_waiting_time}, the agent will be waiting at most $d$ timesteps for this edge to become active.

    In general, after traversing the edge $e_i$, the agent attempts to traverse $e_{i + 1}$, waiting at the current vertex $v_i$ until $e_{i + 1}$ becomes active. From Corollary \ref{col:general_graphs_waiting_time_full}, the edge $e_{i + 1}$ will become active at most $d$ timesteps after the agent arrives at $v_i$. By extension, the agent will wait for at most $d$ timesteps at $v_i$. Therefore, as the walk can contain at most $2 n$ edges, our exploration requires at most $2 d n$ timesteps, giving the theorem.
\end{proof}

\subsection{Lower Bounds}

We now provide some lower bounds on the fastest temporal walk exploring word-representable temporal graphs. Explicitly, we give a construction to provide, for any number of vertices $n$ and diameter $d$, a word-representable graph with $n$ vertices and diameter $d$ requiring $\Omega(d n)$ timesteps to explore.
%
We start by providing a construction for a length $n$ path requiring $\Omega(n^2)$ timesteps to explore, which we use as a gadget to build our main result.

\noindent
\textbf{Path Graph Construction.}
We construct our temporal path graphs as follows. Given some $n \in \mathbb{N}$, we construct a length $n$ path represented by the word $w \in \Sigma^{2 n}$ where $\Sigma = 1, 2, \dots, n$ ordered $1 < 2 < \dots < n$. First, let $w[1, 3] = 1 2 1$. Then, for each $x \in 2, 3, \dots, n - 1$, let $w[2 x] = x + 1$, where $(x + 1)$ is the symbol following $x$ in $\Sigma$, and $w[2 x + 1] = x$. Finally, let $w[2 n] = n$.

\begin{observation}
    \label{obs:positional_value_of_w_k}
    The construction given above constructs a word $w \in \Sigma^{2n}$ where, for any $k \in \mathbb{N}$, $\ell \in [4, 2n - 1]$, $k' \in [0, k - 1]$,
    $w^k[k' \cdot n + \ell] = \begin{cases}
        (\ell/2) + 1 & \text{$k$ is even,}\\
        (\ell - 1)/2 & \text{$k$ is odd.}
    \end{cases}$
\end{observation}

\begin{lemma}
    \label{lem:construction_gives_path}
    The construction given above constructs a word $w \in \Sigma^{2n}$ where $TG(w)$ a temporal path graph.
\end{lemma}

\begin{proof}
    First, we prove that this graph is a path graph. Consider some pair of symbols $x, y \in \Sigma$ where $2 \leq x < y - 1$. As the symbol $x$ only appears in $w$ at positions $2(x - 1)$ and $2x + 1$, and $y$ at positions $2(y - 1)$ and $2y$, then, as $2x + 1 < 2y - 2 $ $\pi_{\{x, y\}}(w) = x x y y$. Hence, $TG(w)$ does not contain, at any timestep, the edge $(v_x, v_y)$. In the other direction, for any pair $x, x + 1$ we have, by construction, that $x$ appears at positions $2x - 2$, and $2 x + 1$, the and $(x + 1)$ at positions $2x$ and $2 x + 3$. Therefore, $\pi_{\{x, y\}}(w) = x (x + 1) x (x + 1)$ and hence $TG(w)$ contains the edge $(v_x, v_{x + 1})$ in at least one timestep. By extension, $TG(w)$ is a temporal graph.
    %
\end{proof}

We now move to determining the structure of the timesteps in these path graphs. The key property of these graphs is that each timestep corresponds to a factor of the word $w$ of length at most $4$. As such, we have a strict bound on the number, and type, of active edges in each timestep. Lemmas \ref{lem:path_timesteps} and Corollary \ref{col:last_timesteps} establishes the structure of the start points in $w$, while Lemma \ref{lem:path_timesteps_general} generalises these results to $w^k$, for any $k \in \mathbb{N}$.

\begin{lemma}
    \label{lem:path_timesteps}
    The construction given above constructs a word $w \in \Sigma^{2n}$ where $TG(w)$ a temporal path graph with $\floor{(2n - 5) / 4}$ timesteps such that the start point of the $i^{th}$ timestep, for any $i \in [2, \floor{(2n - 5) / 4} - 1]$, is $S_i = 4 i - 5$.
\end{lemma}

\begin{proof}
    We prove this in an inductive manner. First, note that as the prefix of $w$ is $1 2 1$, we have that $S_2 = 3$, matching the statement. Assume now that this holds for any $i' \in [2, i - 1]$. Then, following Observation \ref{obs:positional_value_of_w_k},  we have $w[S_{i - 1}] = 2i - 5$, $w[S_{i - 1} + 1] = 2 i - 3$, $w[S_{i - 1} + 2] = 2 i - 4$, $w[S_{i - 1} + 3] = 2i - 2$, and $w[S_{i - 1} + 4] = 2i - 3 = w[S_{i - 1} + 1]$, hence $S_i = S_{i - 1} + 4 = 4i - 5$. To determine the number of timesteps $T$ in $TG(w)$, we have that $S_T \leq 2n < S_{T + 1}$, hence $4T - 5 \leq 2n < 4 T - 1$, thus $T = \floor{(2n - 5) / 4}$.
\end{proof}

\begin{corollary}
    \label{col:last_timesteps}
    The construction given above constructs a word $w \in \Sigma^{2n}$ where $TG(w)$ is a temporal path graph with $\floor{(2n - 5) / 4}$ timesteps such that the start point of the last timestep is $2n - 1$, if $n$ is odd, or $2n - 2$ if $n$ is even.
\end{corollary}

\begin{proof}
    If $n$ is odd, $S_{\floor{(2n - 5) / 4} - 1} = 2n - 3$ thus we have $w[2n - 3, 2n] = (n - 2) (n) (n - 1) (n)$, giving $S_{\floor{(2n - 5) / 4}} = 2n - 1$.
    If $n$ is even, $S_{\floor{(2n - 5) / 4} - 1} = 2n - 5$ thus we have $w[2n - 5, 2n] = (n - 3) (n - 1) (n - 2) (n) (n - 1) (n)$, giving the statement. 
\end{proof}

\begin{lemma}
    \label{lem:path_timesteps_general}
    The construction given above constructs a word $w \in \Sigma^{2n}$ where, for any $k \in \mathbb{N}$, $TG(w^k)$ a temporal path graph with a lifetime of $T_k = k(\floor{(2n - 5) / 4} - 1) + 1$, and such that the start point of the $T_j + i^{th}$ timestep, for any $i \in [2, T_1 - 1]$, is $S_{T_j + i} = 2j n + 4i - 5$, where $T_j$ is the lifetime of $TG(w^j)$ for some $j \in \mathbb{N}$.
\end{lemma}

\begin{proof}
    Let $T_1 = \floor{(2n + 5) / 4}$, and $T_k$ be the lifetime of $TG(w^k)$.
    We consider two cases. First, if $n$ is odd, observe that $w^k[S_{T_1}, S_{T_1 + 5}] = (n - 1) n 1 2 1$ and hence $S_{T_2} = 2n + 3$. Otherwise, if $n$ is even, $w^k[S_{T_1}, S_{T_1 + 3}] = n 1 2 1$ and again $S_{T_2} = 2n + 3$. Then, following the same outline as Lemma \ref{lem:path_timesteps}, we have $S_{T_{j} + i} =  2j n + 4i - 5$. 

    To determine the total lifetime of $TG(w^k)$, observe that, for any $j \in [2, k]$, $2(j - 1) n \leq S_{i} \leq 2jn$ for exactly $T_1 - 1$ values of $S_i$. Thus we have $2(k - 1)(T_1 - 1) + T_1 = 2k(T_1 - 1) + 1$ total timesteps in $TG(w^k)$.
\end{proof}

We now show that any exploration of these temporal path graphs requires $\Omega(n^2)$ timesteps. The key insight is that, following Lemmas \ref{lem:path_timesteps}, and \ref{lem:path_timesteps_general}, alongside Corollary \ref{col:last_timesteps}, each timestep corresponds to some factor of $w^k$ containing exactly four unique symbols, and five active edges. Further, there exists at least two edges in any timestep $E_t$ that are not active in the subsequent timestep $E_{t + 1}$. Thus, any agent exploring the path is left waiting for $\Omega(n)$ timesteps after exploring some constant number of timesteps, in this case, at most five. As such, the temporal walk exploring the path can be summarised as moving along at most five vertices, then waiting for $\Omega(n)$ timesteps to move the next five. This gives a lower bound on the length of any temporal walk exploring the full path of $\Omega(n^2)$, matching the upper bound given in Theorem \ref{thm:exploration_upper_bounds}.

\begin{theorem}
	\label{thm:paths_worst_case}
	Given an $n \in \mathbb{N}$, there exists a word-representable temporal path graph of length $n$ represented by a word  $w$ of length $4n^2$ requiring $\Omega(n^2)$ timesteps to explore.
\end{theorem}

\begin{proof}
    Let $w$ be constructed as above, and consider $w^{n}$. Note that any agent starting at some vertex $v_{i} \in V \setminus\{v_1, v_n\}$ will take strictly longer to explore the path than some agent starting at $v_1$ or $v_n$, as it is necessary for the agent to reach one of these two vertices, then the other. Therefore, the fastest exploration of $TG(w^k)$ will always involve some agent starting at either $v_1$ or $v_n$.

    Now, consider the path between $v_i$ and $v_{i - 3}$.  The agent may only to move from $v_i$ to $v_{i - 4}$ after the edges $(v_{i}, v_{i - 1}), (v_{i - 1}, v_{i - 2}),$ and $(v_{i - 2}, v_{i - 3})$ have been active in 3 sequential timesteps. Let $t$ be the earliest timestep at which the agent may traverse the edge $(v_{i}, v_{i - 1})$. We assume, without loss of generality, that $w^n[S_t, S_{t + 1} - 1] = (i - 4) (i - 2) (i - 3) (i - 1)$. Therefore, 
    \begin{align*}
        &w^n[S_{t + 1}, S_{t + 2} - 1] = (i - 2) (i) (i - 1) (i + 2)\\
        &w^n[S_{t + 2}, S_{t + 3} - 1] = (i + 1) (i + 3) (i + 2) (i + 4).
    \end{align*}
    From this, we see that the agent may traverse $(v_{i - 1}, v_{i - 2})$ in timestep $t + 1$, however, as $(i - 2) \notin \letters(w^n[S_{t + 2}, S_{t + 3} - 1])$ and $(i - 3) \notin \letters(w^n[S_{t + 2}, S_{t + 3} - 1])$, the edge $(v_{i - 2}, v_{i - 3})$ is not active at timestep $t + 3$.
    Following the arguments of Lemmas \ref{lem:path_timesteps} and \ref{lem:path_timesteps_general}, the timestep next timestep at which $(v_{i - 3}, v_{i - 2})$ is active is no earlier than $t + \floor{2n - 5}/4 - 3$. Therefore, if the agent starts at $v_n$, the earliest timestep at which any exploration of $TG(w^n)$ can be completed is $\sum_{i \in [1, n/3]} \floor{2n - 5}/4 - 3 = n \floor{2n - 5}/12 - n$.
    
    In the other direction, consider the path between $v_{i}$ and $v_{i + 5}$. We assume, without loss of generality, that the first timestep $E_t$ at which the agent can move from $v_i$ to $v_{i + 1}$ corresponds to the factor $w^k[S_t, S_{t + 1} - 1](i - 3) (i - 1) (i - 2) (i)$. Now, consider the factors
    \begin{align*}
        &w^k[S_{t + 1}, S_{t + 2} - 1] = (i - 1)(i + 1)(i) (i + 2),\\
        &w^k[S_{t + 2}, S_{t + 3} - 1] = (i + 1)(i + 3)(i + 2) (i + 4),\\
        &w^k[S_{t + 3}, S_{t + 4} - 1] = (i + 3)(i + 5)(i + 4) (i + 6),\\
        &w^k[S_{t + 4}, S_{t + 5} - 1] = (i + 5)(i + 7)(i + 6) (i + 8),\\
        &w^k[S_{t + 5}, S_{t + 6} - 1] = (i + 7)(i + 9)(i + 8) (i + 10)
    \end{align*}
    Note that the agent may reach vertex $v_{i + 1}$ in timestep $t$, $v_{i + 2}$ in timestep $t + 1$, $v_{i + 3}$ in timestep $t + 2$, $v_{i + 4}$ in timestep $t + 3$, and $v_{i + 5}$ in timestep $t + 4$. However, as $(i + 5) \notin \letters(w^k[S_{t + 5}, S_{t + 6} - 1])$ and $(i + 6) \notin \letters(w^k[S_{t + 5}, S_{t + 6} - 1])$, by the arguments above, the agent then must wait $\floor{2n - 5}/4$ timesteps to reach $v_{i + 5}$.
    Therefore, if the agent starts at $v_1$, the earliest timestep at which any exploration of $TG(w^n)$ can be completed is $\sum_{i \in [1, n/5} \floor{2n - 5}/4 - 5 = n \floor{2n - 5}/20$. 

    Therefore, the number of timesteps needed for an agent to travel from $v_1$ to $v_n$ (resp., $v_n$ to $v_1$) is $\Omega(n^2)$.
\end{proof}

\paragraph*{General Lower Bound.}

We now move to our more general lower bound. We assume, without loss of generality, that we are given a pair $n, d \in \mathbb{N}$ such that $d \in [2, n]$, and $n / d \in \mathbb{N}$. We construct a graph, $G = (V, E)$, containing a set of $d$ ``layers'', $L_1, L_2, \dots, L_d$, such that $\vert L_i \vert = k$, $\forall i \in [1, d]$, with the property that, for any pair $v_i \in L_i$, $v_j \in L_j$, $(v_i, v_j) \in E$ iff $i = j - 1$ or $j = i - 1$. We use, as a basis, a set of path graphs constructed as above.

\noindent
\textbf{Construction},
Let $\Sigma = \{(i, j) \mid i \in [1, n / d], j \in [1, d] \}$. We create a set of $n / d$ word-representable temporal path graphs $P_1, P_2, \dots, P_{n / d}$, represented by $w_1, w_2, \dots, w_{n/d}$ with $\letters(w_i) = \{(i, j) \mid j \in [1, d]\}$, and by extension $P_i$ corresponds to the vertex set $\{v_{i, j} \mid j \in [1, d]\}$. We construct a new word $v = u_1 u_2 \dots u_{d}$ in a manner analogous to the construction of each path graph.
Let $u_1 = (1, 1) (2, 1) \dots (n/d, 1)$, $u_2 = 0 (1, 2) (2,2) \dots (n/d, 2)$, and $u_3 = 0 (n/d, 1) (n/d - 1,1) \dots (1, 1)$. For the remaining factors, $u_i$ is either $w_1[i] w_2[i] \dots w_{n/d}[i]$, if $i$ is even, or $w_{n/d}[i] w_{n/d - 1}[i] \dots w_1[i]$ if $i$ is odd.

We start by showing the key structural results of these graphs. Namely, we show that, given any pair of vertices $v_{i, j}, v_{\ell, k}$, there exists, in any timestep $t$, the edge $(v_{i, j}, v_{\ell, k}) \in E_t$ iff $i = \ell - 1$ or $i = \ell + 1$.

\begin{observation}
    \label{obs:alternating_paths}
    Given some word $w$ representing a path as constructed above and symbol $x \in \letters(w)$, let $\xi_1, \xi_2$ be the indices such that $w[\xi_1] = w[\xi_2] = x$ and $\xi_1 < \xi_2$. Then, $\xi_2 - \xi_1 = 3$.
\end{observation}


\begin{lemma}
    \label{lem:layer_graph_avoids_fake_edges}
    Let $v$ be the word constructed as above, and let $TG(v) = (V, E_1, E_2, \dots,E_T)$. Then, for any pair $v_{i_1, i_2}, v_{j_1, j_2} \in V$ where $i_2 \neq j_2 - 1$ and $j_2 \neq i_2 - 1$, $(v_{i_1, i_2}, v_{j_1, j_2}) \notin E_t$ for any $t \in [1, T]$.
\end{lemma}

\begin{proof}
    Assume the contradiction and, without loss of generality, that $i_1 < j_1$ and $i_2 < j_2$. First, let $i_2 = j_2$, and let $\ell_1, \ell_2 \in [1, d]$ be the pair of indices such that $\ell_1 < \ell_2$ and $(i_1, i_2), (j_1, j_2) \in \letters(u[\ell_1]) \cap \letters(u[\ell_2])$. Following Observation \ref{obs:alternating_paths} and the construction, we have that, if $\ell_1$ is even, $\pi_{\{(i_1, i_2), (j_1, j_2)\}} = (i_1, i_2) (j_1, j_2) (j_1, j_2) (i_1, i_2)$, or, if $\ell_1$ is odd, $\pi_{\{(i_1, i_2), (j_1, j_2)\}}(v) = (i_1, i_2) (j_1, j_2) (j_1, j_2) (i_1, i_2)$. Now, consider some case where $i_2 < j_2 - 1$. Then, by construction, $\pi_{\{(i_1, i_2), (j_1, j_2)\}} = (i_1, i_2) (i_1, i_2) (j_1, j_2) (j_1, j_2)$. Thus, in either case, the statement holds.
\end{proof}

\begin{lemma}
    \label{lem:layer_graph_has_real_edges}
    Let $v$ be the word constructed as above and let $TG(v) = (V, E_1, E_2, \dots,E_T)$. Then, for any pair $v_{i_1, i_2}, v_{j_1, j_2} \in V$ where $i_2 = j_2 - 1$ or $j_2 = i_2 - 1$, $\exists t \in [1, T]$ such that $(v_{i_1, i_2}, v_{j_1, j_2}) \in  E_t$.
\end{lemma}

\begin{proof}
    Assume, without loss of generality, that $i_1 < j_1$ and $i_2 = j_2 - 1$. Then, we have that $(i_1, i_2) \in \letters(u_{2 i_2 - 2})$ and $(i_1, i_2) \in \letters(u_{2 i_2 + 1})$. Further, $(j_1, j_2) \in \letters(u_{2 j_2 - 2})$ and $(j_1, j_2) \in \letters(u_{2 j_2 + 1})$. Therefore, as $j_2 = i_2 + 1$, we have that $(i_1, i_2) \in \letters(u_{2 i_2 - 2})$, $(j_1, j_2) \in \letters(u_{2 i_2})$, $(i_1, i_2) \in \letters(u_{2 i_2 + 1})$, and $(j_1, j_2) \in \letters(u_{2 i_2 + 3})$, hence $\pi_{\{i_1, i_2), (j_1, j_2)\}} = (i_1, i_2) (j_1, j_2) (i_1, i_2) (j_1, j_2)$, giving the result.
\end{proof}

\begin{corollary}
    \label{col:general_good_edge_behavior}
    Let $v$ be the word constructed as above, let $k \in \mathbb{N}$ be some natural integer, and let $TG(v^k) = (V, E_1, E_2, \dots,E_T)$. Then,  for any pair $v_{i_1, i_2}, v_{j_1, j_2} \in V$, $\exists t \in [1, T]$ such that $(v_{i_1, i_2}, v_{j_1, j_2}) \in  E_t$ iff $i_2 = j_2 - 1$ or $j_2 = i_2 - 1$.
\end{corollary}

With Lemmas \ref{lem:layer_graph_avoids_fake_edges} and \ref{lem:layer_graph_has_real_edges}, along with Corollary \ref{col:general_good_edge_behavior}, we see an analogy between the path graphs given previously and these new graphs. More specifically, we have that we may transition between at most five layers in any temporal walk before waiting $\Omega(d)$ timesteps. As there does not exist, at any timestep, an edge between any pair of vertices in the same layer, this bound can more generally be applied, showing that we may visit at most five vertices before waiting $\Omega(d)$ time. This gives our final lower bound of $\Omega(d n)$ for exploring any graph in this class.

\begin{theorem}
    \label{thm:ultimate_lower_bound}
    Given any pair $n, d \in \mathbb{N}$ where $d < n$, there exists some word $v \in \Sigma^*$ representing a temporal graph $TG(v)$ with $n$ vertices, and a diameter of $d$, requiring $\Omega(d n)$ timesteps to explore.
\end{theorem}

\begin{proof}
    Let $v$ be constructed as above, let $TG(v^{n}) = (V, E_1, E_2, \dots, E_T) $, and let $L_j = \{v_{i, j} \mid i \in [1, n / d] \}$.
    Observe that if any edge $(v_{i, j}, v_{i + 1, \ell}) \in E_t$ then, $(v_{i, j}, v_{i + 1, \ell'}) \in E_t$ $\forall \ell' \in [1, n / d]$. As, by Lemma \ref{lem:layer_graph_has_real_edges} and Corollary \ref{col:general_good_edge_behavior}, there are no edges between any pair of vertices $v_i, v_{i'} \in L_j$, any temporal walk exploring $TG(v^n)$ involves the agent moving between some pair of layers each time an edge is traversed. Therefore, considering layers in a manner analogous to the vertices in the path graph, following the same arguments as Theorem \ref{thm:paths_worst_case}, the agent can explore at most $5$ vertices in $\floor{2d - 5}/4 - 5$ timesteps, with any temporal path between some pair of vertices in the same layer $L_j$ requiring the agent to move to and from an adjacent layer, either $L_{j - 1}$ or $L_{j + 1}$.
    Hence, exploring all $n$ vertices of the graph requires $\Omega(d n)$ timesteps, completing the proof.
\end{proof}

\section{Conclusion}

In this paper, we have introduced the class of word-representable temporal graphs and provided a set of initial results on the problem of exploring such graphs. Explicitly, we show that any always connected temporal graph can be explored in $O(\delta n)$ timesteps, where $\delta$ is the minimum degree of any vertex in the graph, the any temporal graph represented by a word of length $n(2dn + d)$ can be explored in $O(d n)$ timesteps, where $d$ is the diameter of the graph, and that there exists, for any diameter $d$ and number of vertices $n$, some word representable temporal graph requiring $\Omega(d n)$ timesteps to explore.

There are two immediate directions in which to continue this work. First is to determine if it is possible to determine an optimal exploration schedule for either class of word-representable temporal graphs. Second is to determine some example of an always connected word-representable temporal graph requiring $\Omega(\delta n)$ timesteps to explore.

More generally, there are many questions about what classes of temporal graphs can, be represented by a word. While there exists a natural constraint that, by our definition, the underlying graph must be representable, the restriction on each timestep may lead to interesting results regarding which temporal graphs can be represented at each timestep.

\bibliography{bib.bib}

\end{document}